\documentclass[11pt,onecolumn]{IEEEtran}
\usepackage{graphicx}
\usepackage{amsmath}
\usepackage{nicefrac}
\usepackage{amsfonts}
\usepackage{times}
\usepackage{algorithm}
\usepackage{algorithmic}
\usepackage{amsthm}
\usepackage{multirow}
\usepackage{rotating}
\usepackage{fancyhdr}
\usepackage[usenames]{color}
\newtheorem{prop}{Proposition}
\newtheorem{remark}{Remark}
\newtheorem{lemma}{Lemma}

\newtheorem{theorem}{Theorem}

\newcommand{\cb}[1]{{\color{black} {#1}}}
\newcommand{\calderBold}{}

\newcommand{\secref}[1]{\S\ref{#1}}

\usepackage{nicefrac, graphicx,amsmath,amsthm,amssymb}
\usepackage{bm,amsmath,amssymb,nicefrac}
\usepackage{algorithmic}
\usepackage{algorithm}

\newcommand{\Tr}{\mbox{\rm Tr}}

\newcommand{\mm}{{N}}
\newcommand{\m}{{\mm}}

\newcommand{\nor}{\frac{1}{\sqrt{\mm}}}

\newcommand{\A}{{\Phi}}

\newcommand{\has}{{{\hat{\as}}}}

\newcommand{\n}{{\cal C}}

\newcommand{\as}{{\alpha}}

\newcommand{\ignore}[1]{}

\usepackage[font=footnotesize]{subfig}
\newcommand{\F}{\mathbb{F}_2^m}

\newcommand{\B}{{\A}}

\newcommand{\nc}{\left[\n\right]}

\usepackage{setspace}
\singlespace
\title{\cb{Reed Muller Sensing Matrices and the LASSO}}
\author{Robert Calderbank\thanks{Department of Electrical Engineering and Department of Mathematics. Princeton University. \textit{calderbk@math.princeton.edu}. }
 and
 Sina Jafarpour\thanks{Department of Computer Science. Princeton University. \textit{sina@cs.princeton.edu}.}
 \thanks{The work of R. Calderbank and S. Jafarpour is supported in part by NSF under grant DMS 0701226, by ONR under grant N00173-06-1-G006, and by AFOSR under grant FA9550-05-1-0443.}
}

\begin{document}
\maketitle

\begin{abstract}
\boldmath
\textbf{\cb{
We construct two families of deterministic sensing matrices where the columns are obtained by exponentiating codewords in the quaternary Delsarte-Goethals code $DG(m,r)$. This method of construction results in sensing matrices with low coherence and spectral norm. The first family, which we call Delsarte-Goethals frames, are $2^m$ - dimensional tight frames with redundancy $2^{rm}$. The second family, which we call Delsarte-Goethals sieves, are obtained by subsampling the column vectors in a Delsarte-Goethals frame. Different rows of a Delsarte-Goethals sieve may not be orthogonal, and we present an effective algorithm for identifying all pairs of non-orthogonal rows. The pairs turn out to be duplicate measurements and eliminating them leads to a tight frame.
Experimental results suggest that all $DG(m,r)$ sieves with $m\leq 15$ and $r\geq2$ are tight-frames; there are no duplicate rows. For both families of sensing matrices, we measure accuracy of reconstruction (statistical $0-1$ loss) and complexity (average reconstruction time) as a function of the sparsity level $k$. Our results show that DG frames and sieves outperform random Gaussian matrices in terms of noiseless and noisy signal recovery using the LASSO.
}}
\end{abstract} 

\section{Introduction}
\label{sec:intro}
The central goal of compressed sensing is to capture attributes of a signal using very few measurements. In most work to date, this broader objective is exemplified by the important special case in which the measurement data constitute a vector $f=\A\as+e$, where $\A$ is an $\mm\times\n$ matrix called the \textit{sensing matrix}, $\as$ is a signal in $\mathbb{C}^\n$, that is well-approximated by a $k$-sparse vector (a signal with at most $k$ non-zero entries), and $e$ is additive measurement noise. 

The role of random measurement in compressive sensing (see \cite{CRT1} and \cite{Donoho}) can be viewed as analogous to the role of random coding in Shannon theory. Both provide worst-case performance guarantees in the context of an adversarial signal/error model. In the standard paradigm, the measurement matrix is required to act as a near isometry on all $k$-sparse signals (this is the Restricted Isometry Property or RIP introduced in \cite{CT}). It has been shown that if a sensing matrix satisfies the RIP property then Basis pursuit \cite{CRT1,CRT2} programs can be used to estimate the best $k$-term approximation of  any signal in $\mathbb{C}^\n$, measured in the presence of any $\ell_2$ norm bounded measurement noise \cite{best}. 

It is known that certain probabilistic processes generate sensing matrices that for $k=O(\mm)$ satisfy $k$-RIP with high probability (see \cite{Baraniuk07}). This is significantly different from the best known results for deterministic sensing matrices \cite{Devore} where $k$-RIP is known only for $k=O(\sqrt{\mm})$. We normalize the columns of a sensing matrix to have unit $\ell_2$ - norm and define the worst case coherence $\mu$ to be the maximum absolute value of an inner product of distinct columns. It follows from the Welch bound \cite{stromher} that $\mu \ge O\left(\frac{1}{\sqrt{\mm}}\right)$. When $\mu = O\left(\frac{1}{\sqrt{N}}\right)$ it then follows from the Gerschgorin Circle Theorem \cite{Nowak} that the sensing matrix satisfies $k$-RIP with $k = O\left(\mu^{-1}\right)$. In general however no polynomial-time algorithm is known for verifying that a sensing matrix with the worst-case coherence $\mu$ satisfies $k$-RIP with $k= \Omega\left(\mu^{-1}\right)$.

 The RIP property is not an end in itself. It provides guarantees for a particular method of signal reconstruction, but there is significant interest in structured sensing matrices and alternative reconstruction algorithms. One example is the adjacency matrices of expander graphs \cite{sina, BGIST} where it is known to be impossible to satisfy RIP with respect to the $\ell_2$ norm \cite{negative}. Sparse signal recovery is still possible with Basis Pursuit since the adjacency matrix acts like a near isometry on k-sparse signals with respect to the $\ell_1$ norm. However error estimates are looser than corresponding estimates for random sensing matrices and resilience to measurement noise is limited to sparse noise vectors.
 
 The coherence between rows of a sensing matrix is a measure of the new information provided by an additional measurement. The coherence between columns of a sensing matrix is fundamental to deriving performance guarantees for reconstruction algorithms such as Basis Puruit. There are two fundamental measures of coherence: The worst-case coherence $\mu$ which measures the maximal coherence between the columns of the sensing matrix, and the spectral norm $\|\A\|_2$ which measures 
the maximal coherence between the rows of the frame. The ideal case is when worst case coherence between columns matches the Welch bound $\left(\mu = O\left(\frac{1}{\sqrt{\mm}}\right)\right)$ and different measurements are orthogonal. Then, with high probability a $k$-sparse vector has a unique sparse representation \cite{tropp10}, and this representation can be efficiently recovered using a LASSO program \cite{CP07}. Section~\secref{sec:bg} introduces notation and reviews prior work on the geometry of sensing matrices and the performance of the LASSO reconstruction algorithm.

In this paper we consider sensing matrices based on the $\mathbb{Z}_4$-linear representation of Delsarte Goethals codes. The columns are obtained by exponentiating codewords in the quaternary Delsarte-Goethals code; they are uniformly and very precisely distributed over the surface of an $\mm$-dimensional sphere. Coherence between columns reduces to properties of these algebraic codes. Section~\secref{sec:bg} reviews the construction of Delsarte-Goethals (DG) sets of $\mathbb{Z}_4$-linear quadratic forms which is the starting point for the construction of the corresponding codes; each quadratic form determines a codeword where the entries are the values taken by quadratic form. Section~\secref{sec:dg} introduces Delsarte-Goethals frames and Delsarte-Goethals sieves; the columns of these sensing matrices are obtained by exponentiating DG codewords. We then determine the worst case coherence and spectral norm for these sensing matrices. 

Cand\`es and Plan \cite{CP07} specified coherence conditions under which a LASSO program will successfully recover a $k$-sparse signal when the k non-zero entries are above the noise variance. We use these results to provide an average case error analysis for stochastic noise in both the data and measurement domains. The Delsarte Goethals (DG) sensing matrices are essentially tight frames so that white noise in the data domain maps to white noise in the measurement domain.

Section~\secref{sec:exp} presents the results of numerical experiments that compare DG frames and sieves with random Gaussian matrices of the same size. The SpaRSA package \cite{sparsa} is used to implement the LASSO recovery algorithm in all cases. DG frames and sieves outperform random matrices in terms of probability of successful sparse recovery but reconstruction time for the DG sieve is greater than that for the other sensing matrices. We remark that there are alternative fast reconstruction algorithms that exploit the structure of DG sensing matrices. \cb{The witnessing algorithm proposed in \cite{isit1} requires less storage, provides support-localized detection, and does not require independence among the support entries. On the other hand, LASSO reconstruction tends to be more robust to noise in the data domain.}

\section{Background and Notation}
\label{sec:bg}
{
This Section introduces notation and reviews  the theory of sparse reconstruction.
\subsection{Notation} Given a vector $v=(v_1,\cdots,v_n)$ in $\mathbb{R}^n$, $\|v\|_2$ denotes the Euclidean norm of $v$, and $\|v\|_1$ denotes the $\ell_1$ norm of $v$ defined as $\|v\|_1\doteq \sum_{i=1}^n |v_i|$. We further define $\|v\|_\infty \doteq \max\left\{|v_1|,\cdots,|v_n| \right\}$, and $\|v\|_{\min} \doteq \min\left\{|v_1|,\cdots,|v_n| \right\}$. Also the Hamming weight of $v$ is defined as  $\|v\|_0\doteq\{i: v_i\neq 0\}$. Whenever clear from the context, we drop the subscript from the $\ell_2$ norm. Also $v_{i\rightarrow j}$ denotes the vector $v$ restricted to entries $i,i+1,\cdots,j$, that is $v_{i\rightarrow j}\doteq (v_i,v_{i+1},\cdots,v_j).$

Let $A$ be a matrix with rank $r$. We denote the conjugate transpose of $A$ by $A^\dag$. Let $\boldsymbol{\sigma}=[\sigma_1,\cdots,\sigma_r]$ denote the vector of the singular values of $A$. The spectral norm $\|A\|$ of a matrix $A$ is the largest singular value of $A$: that is $\|A\|\doteq \|\boldsymbol{\sigma}\|_{\infty}.$ The condition number of $\A$ is the ratio between its largest and its smaller singular values: $\varsigma(A)\doteq\frac{ \|\boldsymbol{\sigma}\|_\infty}{ \|\boldsymbol{\sigma}\|_{\min}}.$ Finally the nuclear norm of $A$, denoted as $\|A\|_1$ is the $\ell_1$ norm of the singular value vector $\boldsymbol{\sigma}$.

Throughout this paper we shall use the notation $\varphi_j$ for the $j^{th}$ column of the sensing matrix $\A$; its entries will be denoted by $\varphi_j(x)$, with the row label $x$ varying from $0$ to $\mm-1$. In other words, $\varphi_j(x)$ is the entry of $\A$ in row $x$ and column $j$. We denote the set $\{1,\cdots.\n\}$ by $\nc$. Let $S$ be a subset of $\nc$. $\A_S$ is obtained by restricting $\A$ to those columns that are listed in $S$.

A vector $\as\in\mathbb{R}^\n$ is $k$-sparse if it has at most $k$ non-zero entries. The support of the $k$-sparse vector $\as$, denoted by $\mbox{Supp}(\as)$, contains the indices of the non-zero entries of $\as$.  Let $\pi=\{\pi_1,\cdots,\pi_\n\}$ be a uniformly random permutation of $\nc$. In this paper, our focus is on the average case analysis, and we always assume that $\as$ is a $k$-sparse signal with $\mbox{Supp}(\as)=\{\pi_1,\cdots,\pi_k\}$. We further assume that conditioned on the support, the values of the $k$ non-zero entries of $\as$ are sampled from a distribution which is absolutely continuous with respect to the Lebesgue measure on $\mathbb{R}^k$.
}
{
\subsection{Incoherent Tight Frames}
An $\mm\times \n$ matrix $\A$ with normalized columns is called a dictionary. A dictionary is a tight-frame with redundancy $\frac{\n}{\mm}$ if for every vector $v\in \mathbb{R}^\n$, $\|\A v\|^2=\frac{\n}{\mm}\,\|v\|^2$. If $\A\A^\dag=\frac{\n}{\mm}{\rm I}_{\mm\times\mm}$, then $\A$ is a tight-frame with redundancy $\frac{\n}{\mm}$ (see \cite{strip}).

\begin{prop}
\label{bg:prop1}
Let $\A$ be an $\mm\times\n$ dictionary. Then $\|\A\|^2\geq \frac{\n}{\mm}$, and equality holds if and only if $\A$ is a tight frame with redundancy $\frac{\n}{\mm}$.
\end{prop}
\begin{IEEEproof}
Let  Let $\boldsymbol{\sigma}$ be the singular value vector of $\A$. We have
\begin{equation}\label{bg:eq1} \|\A\|^2= \|\boldsymbol{\sigma}\|_{\infty}^2 \geq \frac{1}{\mm} \sum_{i=1}^{\mm} \sigma_i^2=  \frac{1}{\mm} {\rm Tr}\left(\A\A^\dag\right)=\frac{\n}{\mm}.
\end{equation}
The inequality in Equation~\eqref{bg:eq1} changes to equality if and only if all the eigenvalues of $\A\A^\dag$  are equal to $\frac{\n}{\mm}$. This is equivalent to the requirement $\A\A^\dag=\frac{\n}{\mm}{\rm I}_{\mm\times\mm}.$
\end{IEEEproof}
}
{
The mutual coherence between the columns of an $\mm\times \n$ sensing matrix is defined as 
\begin{equation}
\label{bg:eq2}
\mu\doteq \max_{i\neq j} \left|\varphi_i^\dag \varphi_j\right|.
\end{equation}
Strohmer and Heath \cite{stromher} showed that the mutual coherence of any $\mm\times \n$ dictionary is at least $\frac{1}{\sqrt{\mm}}$. Designing dictionaries with small spectral norms (tight frames in the ideal case), and with small coherence $\left(\mu=O\left(\frac{1}{\sqrt{\mm}}\right)\mbox{ in the ideal case}\right)$ is useful in compressed sensing for the following reasons.}
{\\
 \textbf{Uniqueness of Sparse Representation ($\boldsymbol{\ell_0}$ minimization)}
The following results are due to Tropp \cite{tropp10} and show that with overwhelming probability the $\ell_0$ minimization program successfully recovers the original $k$-sparse signal.
\begin{theorem}
\label{bg:prop1-2}
 Assume the dictionary $\A$ satisfies $\mu\leq \frac{c}{\log \n}$, where $c$ is an absolute constant. Further assume $k \leq \frac{c\, \n}{\|\A\|^2 \log\n}$. Let $S$ be a random subset of $\nc$ of size $k$, and let $\A_S$ be the corresponding $\mm\times k$ submatrix. Then there exists an absolute constant $c_0$
$$\Pr\left[\left\|\A_S^\dag\A_S-I\right\|\geq c_0\left(\mu \log \n + 2\sqrt{\frac{\|\A\|^2\,k}{\n}} \right) \right]\leq 2\, \n^{-1}.
$$
\end{theorem}

\begin{theorem}
\label{bg:prop2}
Assume the dictionary $\A$ satisfies $\mu\leq \frac{c}{\log \n}$, where $c$ is an absolute constant. Further assume $k \leq \frac{c\, \n}{\|\A\|^2 \log\n}$. Let $\as$ be a $k$-sparse vector, such that the support of the $k$ nonzero coefficients of $\as$ is selected uniformly at random. Then with probability $1-O\left(\n^{-1}\right)$ $\as$ is the unique $k$-sparse vector mapped to $u= \A\as$ by the measurement matrix $\A$.
\end{theorem}
\textbf{Sparse Recovery via LASSO ($\boldsymbol{\ell_1}$ minimization)}
Uniqueness of sparse representation is of limited utility given that $\ell_0$ minimization is computationally intractable. However, given modest restrictions on the class of sparse signals, Cand\`es and Plan  \cite{CP07} have shown that with overwhelming probability the solution to the $\ell_0$ minimization problem coincides with the solution to a convex lasso program.

\begin{theorem}
\label{bg:prop3} Assume the dictionary $\A$ satisfies $\mu\leq \frac{c}{\log \n}$, where $c$ is an absolute constant. Further assume $k \leq \frac{c_1\, \n}{\|\A\|^2 \log\n}$, where $c_1$ is a numeric constant. Let $\as$ be a $k$-sparse vector, such that 
\begin{enumerate}\item The support of the $k$ nonzero coefficients of $\as$ is selected uniformly at random.\item Conditional on the support, the signs of the nonzero entries of $\as$ are independent and equally likely to be $-1$ or $1$. 
\end{enumerate}
Let $u=\A\as+e$, where $e$ contains $\mm$ iid ${\cal N}(0,\sigma^2)$ Gaussian elements.  Then if $\|\as\|_{\min}\geq 8\sigma\,\sqrt{2\log\n}$, with probability $1-O(\n^{-1})$ the lasso estimate $$\as^*\doteq\arg \min_{\as^+\in \mathbb{R}^{\n}}\frac{1}{2} \|u-\A\as^+\|^2  + 2\,\sqrt{2\log \n}\,\sigma^2\, \|\as^+\|_1$$ has the same support and sign as $\as$, and $\|\A\as-\A\as^*\|^2\leq c_2\, k \,\sigma^2$, where $c_2$ is a numeric constant.
\end{theorem}

 \textbf{Stochastic noise in the data domain.}
 The tight-frame property of the sensing matrix makes it possible to map iid Gaussian noise in the data domain to iid Gaussian noise in the measurement domain:
\begin{lemma}\label{bg:lemma1}
Let $\varepsilon$ be a vector with $\n$ iid ${\cal N}(0,\sigma_d^2)$ entries and $e$ be a vector with $\m$ iid ${\cal N}(0,\sigma_m^2)$ entries. Let $\hbar=\A\varepsilon$ and $\nu=\hbar+e$. Then $\nu$ contains $\mm$ entries, sampled iid from ${\cal N}\left(0, \sigma^2\right)$, where $\sigma^2=\frac{\n}{\mm}\sigma_d^2+\sigma_m^2$.
\end{lemma}
\begin{IEEEproof}
The tight frame property implies
 $$\mathbb{E}\left[\hbar\hbar^\dag\right]=E[\A\varepsilon\varepsilon^\dag\A^\dag]=\sigma_d^2 \A\A^\dag=\frac{\n}{\mm}\sigma_d^2\,I.$$
  Therefore, $\nu=\hbar+e$ contains iid Gaussian elements with zero mean and variance $\sigma^2$. 
\end{IEEEproof}
}
{
Next we construct two families of low-coherence tight frames from Delsarte-Goethals codes.
}

\subsection{Delsarte-Goethals Sets of Binary Symmetric Matrices}
\label{dg_set}
The finite field {\calderBold $\mathbb{F}_{2^m}$} is obtained from the binary field {\calderBold $\mathbb{F}_{2}$} by adjoining a root $\xi$ of a primitive irreducible polynomial  $g$ of degree $m$. The {\calderBold elements} of  $\mathbb{F}_{2^m}$ are  polynomials in $\xi$ of degree at most $m-1$ with coefficients in $\mathbb{F}_{2}$, and we will identify the polynomial $x_0+x_1\xi+\cdots+x_{m-1}\xi^{m-1}$ with the binary $m$-tuple $\left(x_0,\cdots,x_{m-1}\right).$ The \textit{Frobenius map} $f:\mathbb{F}_{2^m}\rightarrow \mathbb{F}_{2^m}$ is defined by $f(x)=x^2$ and the {\calderBold \textit{Trace map}} $\Tr: \mathbb{F}_{2^m} \rightarrow \mathbb{F}_{2}$ is defined by $$\Tr(x)\doteq x+x^2+\cdots+x^{2^{m-1}}.$$ 

The identity $(x+y)^2 = x^2 + y^2$ implies that $\Tr(x+y) = \Tr(x) + \Tr(y)$; the trace is a linear map over the binary field $\mathbb{F}_2$. The trace inner product given by 
$(v,w) =\Tr(vw)$ is non-degenerate; if $\Tr(vz)=0$ for all $z$ in $\mathbb{F}_2^m$ then $v=0$. Every element $a$ in $\mathbb{F}_{2^m}$ determines a symmetric bilinear form $\Tr[xya]$ to which is associated a binary symmetric matrix $P^0(a)$.
$$\Tr[xya]\doteq(x_0\cdots x_{m-1})P^0(a)(y_o\cdots y_{m-1})^\top.$$
{\calderBold The \textit{Kerdock set} $\boldsymbol{K_m}$ } is the $m$-dimensional binary vector space formed by the matrices $P^0(a)$. For example, let $m=3$, and assume the finite field $\mathbb{F}_{8}$ is generated by adjoining a root $\xi$ of the polynomial $g(x)=x^3+x+1$. Then $K_3$ is spanned by{ 
$$P^0(100)=\left(\begin{array}{ccc}1  & 0  & 0  \\0  & 0  & 1  \\0 & 1  & 0  \end{array}\right),~~P^0(010)=\left(\begin{array}{ccc}0  & 0  & 1  \\0  & 1  & 0  \\1 & 0  & 1  \end{array}\right),~\mbox{and }P^0(001)=\left(\begin{array}{ccc}0  & 1  & 0  \\1  & 0  & 1  \\0 & 1  & 1  \end{array}\right)$$}
\begin{theorem}Every nonzero matrix in $K_m$ is nonsingular.\end{theorem}
\begin{proof} If $xP^0(a)=0$ then $\Tr[xya]=0$ for all $y\in \mathbb{F}_{2^m}$. Now the non-degeneracy of the trace implies $a=0$.
\end{proof}

Next we define higher order bilinear forms, each associated with a binary symmetric matrix. Given a positive integer $t$ where $0<t<\frac{m-1}{2}$ and given a field element $a$ $$\Tr\left[\left(xy^{2^t}+x^{2^t}y\right)a\right]$$ defines a symmetric bilinear form that is represented by a binary symmetric matrix $P^t(a)$ as above:
\begin{equation}\label{field_eq}\Tr\left[\left(xy^{2^t}+x^{2^t}y\right)a\right]\doteq(x_0\cdots x_{m-1})P^t(a)(y_o\cdots y_{m-1})^\top\end{equation}
{\calderBold The \textit{Delsarte-Goethals set} ${DG(m,r)}$} is then defined as 
$$DG(m,r)\doteq\left\{ \sum_{t=0}^r P^t(a_t)\,|\,a_t \in \mathbb{F}_{2^m},~t=0,1,\cdots,r\right\}.$$
The Delsarte-Goethals sets are nested
$$K_m=DG(m,0)\subset DG(m,1)\subset \cdots \subset DG\left(m,\frac{m-1}{2}\right),$$
and every bilinear form is associated with some matrix in $DG\left(m,\frac{m-1}{2}\right).$

For example, let $m=3$ and $~g(x)=x^3+x+1$, the set $DG(3,1)$ is spanned by $K_3$, and {
$$P^1(100)=\left(\begin{array}{ccc}0  & 0  & 0  \\0  & 0  & 1  \\0 & 1  & 0  \end{array}\right),~~P^1(010)=\left(\begin{array}{ccc}0  & 1  & 0  \\1 & 0  & 0  \\0 & 0  & 0  \end{array}\right),~\mbox{and }P^1(001)=\left(\begin{array}{ccc}0  & 1  & 1  \\1  & 0  & 0  \\1 & 0  & 0  \end{array}\right).$$}
\begin{theorem}  Every nonzero matrix in $DG(m,r)$ has rank at least $m-2r$.
\end{theorem}
\begin{proof} If $x$ is in the null space of $\sum_{t=0}^r P^t(a_t)$, then for all $y\in \mathbb{F}_{2^m}$
$$\Tr\left[ xya_0+\sum_{t=1}^r \left(xy^{2^t}+x^{2^t}y\right)a_t\right]=0.$$
Since $\Tr(x)=\Tr(x^2)=\cdots=\Tr\left(x^{\frac{1}{2}}\right)$ we have
$$\Tr\left[ \left(\left(xa_0\right)^{2^r}+\sum_{t=1}^r \left(xa_t\right)^{2^{t-r}}+a_t^{2^r}\,x^{2^{t+r}}\right)y^{2^r}\right]=0.$$
Non-degeneracy of the trace now implies
$$\left(xa_0\right)^{2^r}+\sum_{t=1}^r \left(xa_t\right)^{2^{t-r}}+a_t^{2^r}\,x^{2^{t+r}}=0.$$
This is a polynomial of degree at most $2^{2r}$ so there are at most $2^{2r}$ solutions. Hence the rank of the binary symmetric matrix $\sum_{t=0}^r P^t(a_t)$ is at least $m-2r$.
\end{proof}

\section{Delsarte-Goethals Sensing}
\label{sec:dg}
\subsection{Delsarte-Goethals Frames}
{
We start by picking an odd number $m$. The $2^m$ rows of the sensing matrix $\B$ are indexed by the binary $m$-tuples $x$, and the $2^{(r+2)m}$ columns are indexed by the pairs $P,b$, where $P$ is an $m\times m$ binary symmetric matrix in the Delsarte-Goethals set $DG(m,r)$, and $b$ is a binary $m$-tuple. The entry $\varphi_{P,b}(x)$ is given by 
\begin{equation}
\label{dg:eq1}
\varphi_{P,b}(x)=\nor\imath^{xPx^\top+2bx^\top}
\end{equation}
 Note that all arithmetic in the expressions $xPx^\top+2bx^\top$ takes place in the ring of integers modulo $4$. Given $P,b$ the vector $xPx^\top +2bx^\top$ is a codeword in the Delsarte-Goethals code (defined over the ring of integers modulo $4$). For a fixed matrix $P$, the $2^m$ columns $\varphi_{P,b}~,~b\in \mathbb{F}_2^m$ form an orthonormal basis. The name Delsarte-Goethals frame (DG frame) reflects the fact that $\A$ is a union of orthonormal bases. Hence, it is a tight-frame with redundancy $\frac{\n}{\mm}$. Delsarte-Goethals frames are highly incoherent (see \cite{strip}):

\begin{prop}
\label{dg:prop3}
Let $m$ and $r$ be non-negative integers where $m$ is odd and $r<\frac{m-1}{2}$. Then the worst case coherence $\mu$ of the sensing matrix derived from the $DG(m,r)$ set satisfies  $\mu\leq \frac{1}{\mm^{\frac{1}{2}-\frac{r}{m}}}$.
\end{prop}
Sensing matrices derived from Delsarte-Goethals sets are incoherent tight frames so the results of Section~\secref{sec:bg} can be brought to bear. The $\mm\times\mm^2$ sensing matrix derived from the Kerdock set is the union of $\mm$ mutually unbiased bases and the worst case coherence matches the lower bound derived by Levenshtein \cite{L} (see also Strohmer and Heath \cite{stromher}).
}

\subsection{Delsarte-Goethals Sieves}
Chirp Detection \cite{strip} and Witness Averaging \cite{quad} are fast reconstruction algorithms that exploit the structure of Delsarte-Goethals frames. By sieving the testimony of witnesses \cite{quad} it is possible to detect the presence or absence of a signal at any given position in the data domain without explicitly reconstructing the entire signal.

There is however an aliasing problem with DG frames. When two signals modulate columns in the same orthonormal basis, spurious tones are generated by both the chirp detection and witness interrogation algorithms. This can be resolved by decimating the DG frame so that no two columns share the same binary symmetric matrix $P$. The simplest way to do this is to retain columns 
 \begin{equation}
\label{subsampled}
\varphi_{P}(x)=\nor\imath^{xPx^\top}. \end{equation}
for which $b=0.$ We call these subsampled matrices Delsarte-Goethals sieves ($DG(m,r)$ sieves) since it is still possible to sieve the testimony of witnesses. Note that each column of a DG sieve, is a column of the corresponding DG sieve, and the worst case coherence bound follows from Proposition~\ref{dg:prop3}. Figure~\ref{dg:fig1} shows the distribution of the absolute value of pairwise inner products between columns of the $DG(5,1)$ sieve. All entries on the main diagonal are equal to $1$, and around the the diagonal there are squares corresponding to translates of the Kerdock set $K_m$.

Table~\ref{norms:tab1} shows that subsampling may increase the spectral norm. This will make it more difficult to reconstruct the signal either by chirp detection or by sieving the testimony of witnesses. We need to understand this increase in order to be able to apply the results of Section~\secref{sec:bg}.

\begin{figure*}[ht]
\centerline{
\subfloat[Inner product between the first $512$ columns of the $DG(5,1)$ matrix]{\includegraphics[width=2.5in]{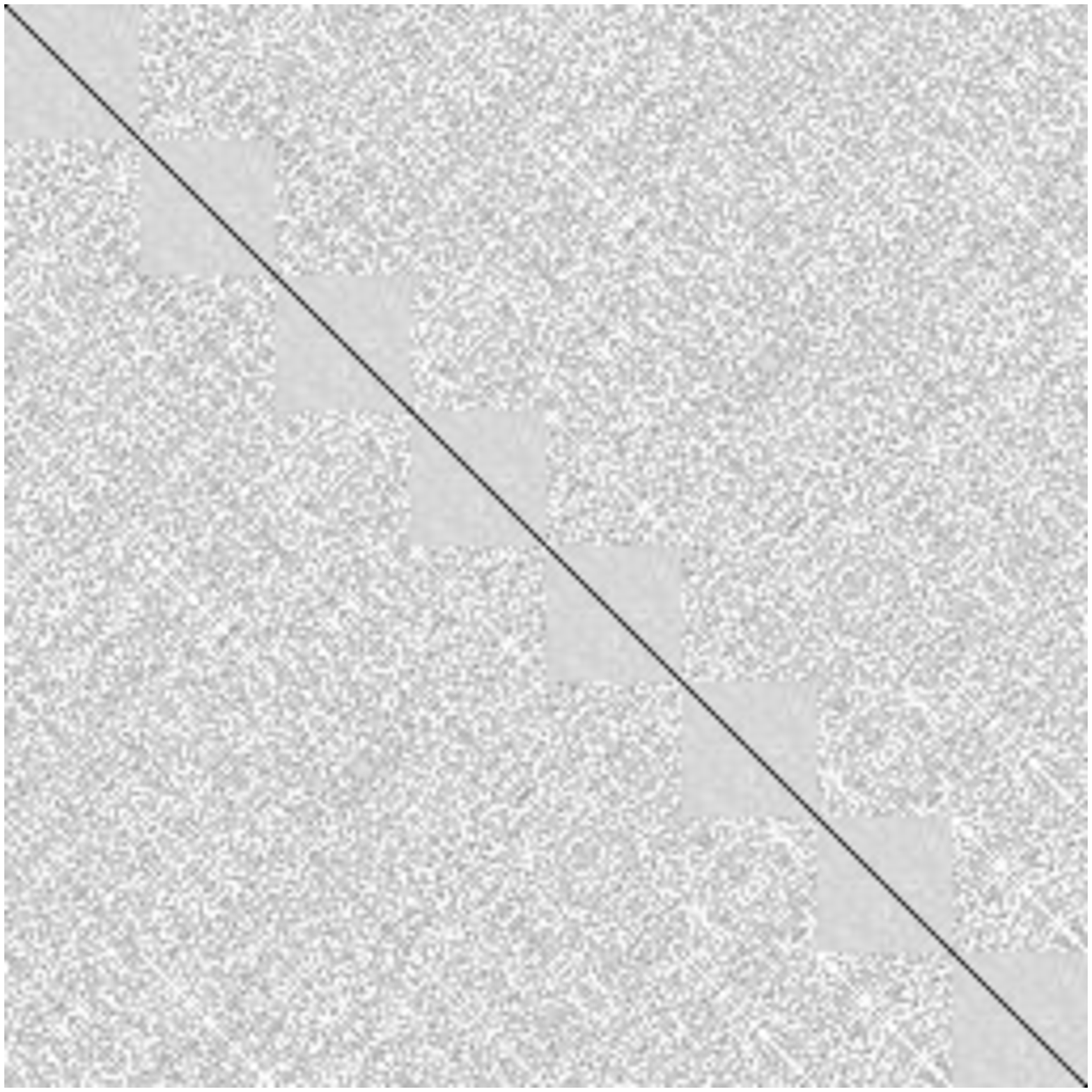}
\label{dg:fig1_1}
}
\hfil
\subfloat[Inner product between the first $256$ columns of the $DG(5,1)$ matrix]{\includegraphics[ width=2.5in]{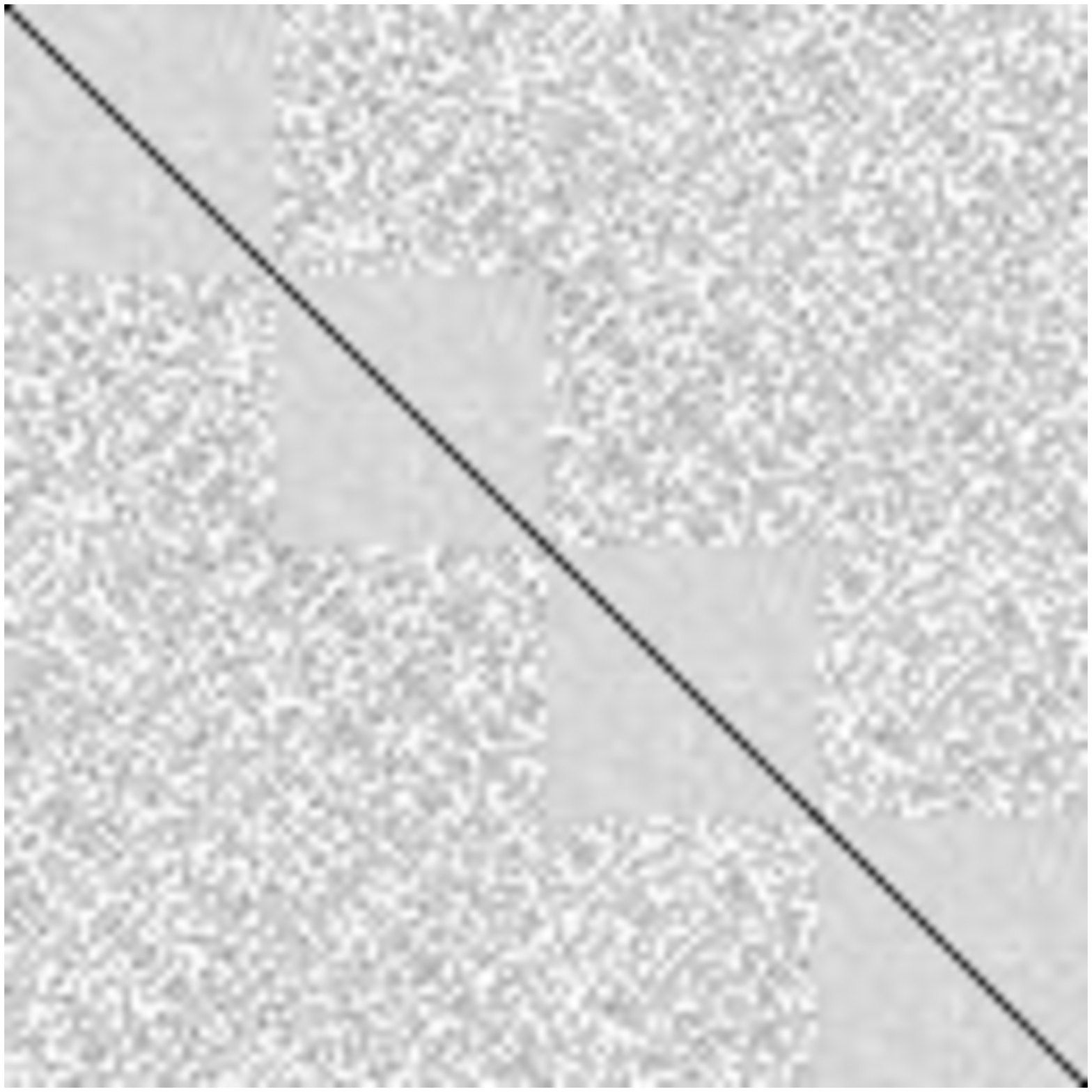}
\label{dg:fig1_2}
}
}
\caption{The inner product between the columns of a $DG(5,1)$ matrix. The point at position $(i,j)$ shows the inner product between the columns $\varphi_i$ and $\varphi_j$. Lighter color shows higher inner product value.}
\label{dg:fig1}
\end{figure*}

\subsection{Spectral Norm of DG Matrices }
\label{dg_norm}
\cb{
\begin{table*}[!t]
\renewcommand{\arraystretch}{1.3}
\caption{Spectral norms of $DG(m,1)$ frames and $DG(m,1)$ sieves as a function of m}
\label{norms:tab1}
\centering
\begin{tabular}{|c|c|c|c|c|}
\hline
\bfseries $DG(m,1)$&$m=3$ & \bfseries $m=5$ & \bfseries $m=7$ & \bfseries$m=9$  \bfseries\\
\hline\hline
 Frame & $2.8284$ &  $5.6569$ & $11.3137$ &  $22.6274$
\\
\hline
Sieve & $5.6568$   &$11.1295$    &$25.0386$ &$55.0338$
\\
\hline
\end{tabular}
\end{table*}
}

Given a sensing matrix, the results presented in Section~\secref{sec:bg} show that if the the worst case coherence and spectral norm are sufficiently small then $\ell_0$ minimization has a unique solution which coincides with the solution of a convex LASSO program. The worst case coherence $\mu$ of the initial $DG(m,r)$ frame satisfies 
$\mu \leq \mm^{\frac{r}{m}-\frac{1}{2}}$. To make sure that every row sum vanishes, we further exclude the $m+1$ rows, indexed by powers of $2$, from the DG sieve. This exclusion changes the worst case coherence by at most $\frac{m+1}{\mm}$ $\left(\mbox{Now }\mu \leq \mm^{\frac{r}{m}-\frac{1}{2}}+\frac{m+1}{\mm}\right)$. The experimental results presented below suggest that the number of pairs of rows in a DG sieve that fail to be orthogonal is very small. Removing these rows results in an equiangular tight frame that is not a union of orthonormal bases.

Table~\ref{norms:tab1} lists the spectral norm of $DG(m,r)$ frames and $DG(m,r)$ sieves for $m= 3,5,7$ and $9$. The spectral norm of a sieve is almost twice that of the corresponding frame and we shall see that the reason is a small number of duplicate rows. Removing these rows results in an equiangular tight frame. We now describe how to find these duplicate rows. 

{Let $x,y$ be two distinct elements of the finite field $\F$, and let $\varphi(x)$, $\varphi(y)$ denote the two rows in $\A$ indexed by $x$ and $y$. Setting $y=x+e$ we obtain
\begin{eqnarray}
\label{th:eq1}
\varphi(x)^\dag\,\varphi(y)&=&\frac{1}{\mm} \sum_{P\in DG(m,r)} \imath^{(x+e)P(x+e)^\top - xPx^\top}
= \frac{1}{\mm} \sum_{P\in DG(m,r)} \imath^{2ePx^\top + ePe^\top}
\\\nonumber &=&  \frac{1}{\mm}  \prod_{t=0}^r \left(\sum_{a\in\F} \imath^{2eP^t(a)xT\top + eP^t(a)eT\top}\right).
\end{eqnarray}
If rows $\varphi(x)$ and $\varphi(y)$ are not orthogonal then each term in the product is nonzero. When $t>0$ we now show that the $t^{th}$ term in the product is a sum of linear characters. Since the index of summation ranges over the group, the sum is either zero or the linear character is trivial (each term in the sum is equal to 1).

\begin{lemma}
\label{th:lemma1}
Let $t\geq 1$ and let $x$ and $x+e$ be two distinct elements of $\F$. Then either $\sum_{a\in\F} \imath^{e\,P^t(a)(2x+e)^\top}$ is zero, or for every field element $a$: $(x+e)P^t(a)(x+e)^\top-xP^t(a)x^\top=0\,\,(\mod 4)$.
\end{lemma}
\begin{IEEEproof}
When $t>0$ every matrix $P^t(a)$ has zero diagonal and the map $a\rightarrow (e+2x) P^t(a)e^\top$ is a linear map from the additive group $\mathbb{F}_2^m$ to $2\,\mathbb{Z}_4$. If this map is not identically zero then the character sum vanishes.
\end{IEEEproof}
The next proposition follows from non-degeneracy of the trace.
\begin{prop}
\label{th:prop1}
If $t>0$ then for every field element $f$
\begin{equation}
\label{th:eq2} f\,P^t(a)f^\top=2\Tr\left(f^{2^t+1}\,a\right)+2z_a f^\top\,(\mbox{ mod }4)\quad \mbox{where } z_a=\left[\Tr\left(\xi^{j(2^t+1)}\,a \right)\,\,j=0,\cdots,m-1 \right].
\end{equation}
\end{prop}
\begin{IEEEproof}
Since the quadratic forms $f P^t(a) f^\top$ and $2\Tr\left( a f^{2^t+1}\right)$ determine the same bilinear form they differ by a linear function 
$2z_a\, f^\top$. Since the quadratic form $f P^t(a) f^\top$ vanishes at all standard coordinate vectors we are able to determine the entries of the vector $2z_a$ that describes the linear function.
\end{IEEEproof} 
Next we use non-degeneracy of the trace to find duplicate rows $\varphi(x)$ and $\varphi(x+e)$.
\begin{lemma}
\label{th:lemma2}
The existence of field elements $x,e$ such that
\begin{equation}
\label{th:eq3}
 (x+e)P^t(a)(x+e)^\top-xP^t(a)x^\top =0\,(\mbox{mod }4)\quad\mbox{for all }a \mbox{ in }\mathbb{F}_2^m, 
\end{equation}
is equivalent to the existence of a solution $\frac{x}{e}$ to the equation
\begin{equation}
\label{th:eq4}
1+\frac{x}{e}+\left(\frac{x}{e}\right)^{2^t}+\sum_{j=0}^{m-1} e_j \left(\frac{\xi^j}{e} \right)^{2^t+1}=0.
\end{equation}
\end{lemma}
\begin{IEEEproof}
Since the trace is a linear map we may replace~\eqref{th:eq3} by the condition that for all $a$ in $\mathbb{F}_2^m$
$$\Tr\left[a\left((x+e)^{2^t+1}+x^{2^t+1}+  \sum_{j=0}^{m-1}e_j \xi^{j\left(2^t+1\right)} \right) \right]=0.$$
Now the non-degeneracy of the trace implies that $(x+e)^{2^t+1}+x^{2^t+1}+  \sum_{j=0}^{m-1}e_j \xi^{j\left(2^t+1\right)}=0$. Expanding $(x+e)^{2^t+1}$, we ontain 
$$e^{2^t+1}+x\,e^{2^t}+x^{2^t}\,e+ \sum_{j=0}^{m-1}e_j \xi^{j\left(2^t+1\right)}=0.$$ Since $e$ is non-zero, dividing the equation by $e^{2^t+1}$ completes the proof.
\end{IEEEproof}
The solutions to the equation $z + z^{2^t} =0$ form a subfield of $\mathbb{F}_2^m$ and the number of solutions is $\gcd \left(2^t-1,2^m-1\right)$. Note that when m is odd and $t=1$ or $t=2$, there are exactly two solutions ($z=0$ and $z=1$).
We now list the conditions satisfied by $x$ and $e$ if the row $\varphi(x)$ is not orthogonal to the row $\varphi(x+e)$.
\begin{theorem}
\label{th:theorem1}
Let $x$ and $x+e$ be two distinct elements of the finite field $\F$. Then $\varphi(x)^\dag \varphi(x+e)\neq 0$ if and only if the following conditions simultaneously hold:
\begin{itemize}
\item(C1) For every $t\geq 1$: $\frac{x}{e}+\left(\frac{x}{e}\right)^{2^t}=1+\sum_{j=0}^{m-1} e_j \left(\frac{\xi^j}{e} \right)^{2^t+1}.$
\item(C2) $\sum_{a\in\F} \imath^{e\,P^0(a)(2x+e)^\top}\neq 0$.
\end{itemize}
\end{theorem}

Theorem~\ref{th:theorem1} provides an efficient way for identifying the non-orthogonal rows of the sieve matrices  without requiring to calculate the gram matrices $\A^\dag\A$ explicitly. For every element $e$, we first find the solution for the case $t=1$. If such a solution exists then we just need to \textit{check} that condition (C1) is valid for other values of $t$. If all conditions passed then we just verify condition (C2). This method significantly reduces the computational cost of eliminating the non-orthogonal rows.

The next formula is for $t=1$
$$\frac{x}{e}+\left(\frac{x}{e}\right)^2=\lambda\quad\mbox{where }\lambda=1+\frac{\sum_{j=0}^{m-1}e_j \xi^{3j}}{e^3}.$$
This is a quadratic equation with roots $\frac{x}{e}$ and $\frac{x}{e} +1$ where $\frac{x}{e}\doteq\sum_{\substack{\ell:\,odd\\1\leq\ell\leq m-2 }} \lambda^{2^\ell}.$ On the other hand $$\lambda+\lambda^2= \frac{x}{e}+\left(\frac{x}{e}\right)^4=\alpha\quad\mbox{where }\alpha=1+\frac{\sum_{j=0}^{m-1}e_j \xi^{5j}}{e^5}.$$ Thus we can also retrieve the explicit solution
$\lambda=\sum_{\substack{\ell:\,odd\\1\leq\ell\leq m-2 }} \alpha^{2^\ell}.$ In other words, the following equivalence between the two field elements (which are both functions of $e$) must be satisfied:
\begin{equation}
\label{th:eq5}
\sum_{\substack{\ell:\,odd\\1\leq\ell\leq m-2 }} \left(1+\frac{\sum_{j=0}^{m-1}e_j \xi^{5j}}{e^5} \right)^{2^\ell}= 1+\frac{\sum_{j=0}^{m-1}e_j \xi^{3j}}{e^3}.
\end{equation}

\cb{
\begin{remark}
Solutions to condition~(C1) correspond to codewords of weight $2$ in the binary code that is dual to the code determined by matrices in $DG(m,r)$ with zero diagonal. The number of solutions can be calculated using the MacWilliams Identities and we provide details in Appendix~\secref{sec:mw}.
\end{remark}
}   
\cb{

 Table~\ref{norms:tab2} records the number of duplicate measurements that need to be deleted in order to transform a $DG(m,1)$ sieve into a tight frame. We calculated the number of duplicate rows for $DG(m,2)$, where $m \leq 15$, and found that there were no solutions to~\eqref{th:eq5} that also satisfied~(C2); that is all $DG(m,2)$ sieves with $m \leq 15$ are tight frames. Hence
 \begin{center}\textbf{Conjecture:} Every $DG(m,r)$ sieve with $r\geq 2$ is a tight-frame.\end{center}
 }

}

Figure~\ref{norms:fig1} displays for $m=7$ and $9$ the average condition number of a random $\mm\times k$ submatrix of the $DG(m,1)$ sieve and the $DG(m,0)$ frame. The spectral norm of the hollow gram matrix $\|\A^\dag\A-{\rm I}_{\mm}\|_2$ was calculated for $2000$ randomly chosen submatrices $\A_k$ and the average was recorded. The comparison with Gaussian sensing matrices was made by drawing $10$ iid Gaussian matrices, calculating for each matrix the average spectral norm over randomly chosen submatrices, and then recording the median value.

\cb{
\begin{table*}[ht]
\renewcommand{\arraystretch}{1.3}
\caption{Number of row deletions required to transform a $DG(m,1)$ sieve into a tight frame.}
\label{norms:tab2}
\centering
\begin{tabular}{|c|c|c|c|c|c|c|}
\hline
\bfseries $DG(m,1)$&$m=5$ & \bfseries $m=7$ & \bfseries $m=9$ & \bfseries$m=11$ & \bfseries$m=13$ & \bfseries$m=15$  \bfseries\\
\hline\hline
 \small{$\#$ of non-orthogonal rows} & $11$ &  $25$ & $45$ &  $83$ &  $203$ &  $381$ 
\\
\hline
 \small{$\%$ of non-orthogonal rows}  & $0.3438$ &  $0.1953$ &  $0.0879$ & $0.0405$ & $0.0248$ & $0.0116$
\\
\hline
\end{tabular}
\end{table*}
}

\begin{figure*}[ht]
\centerline{
\subfloat{\includegraphics[width=3.5in]{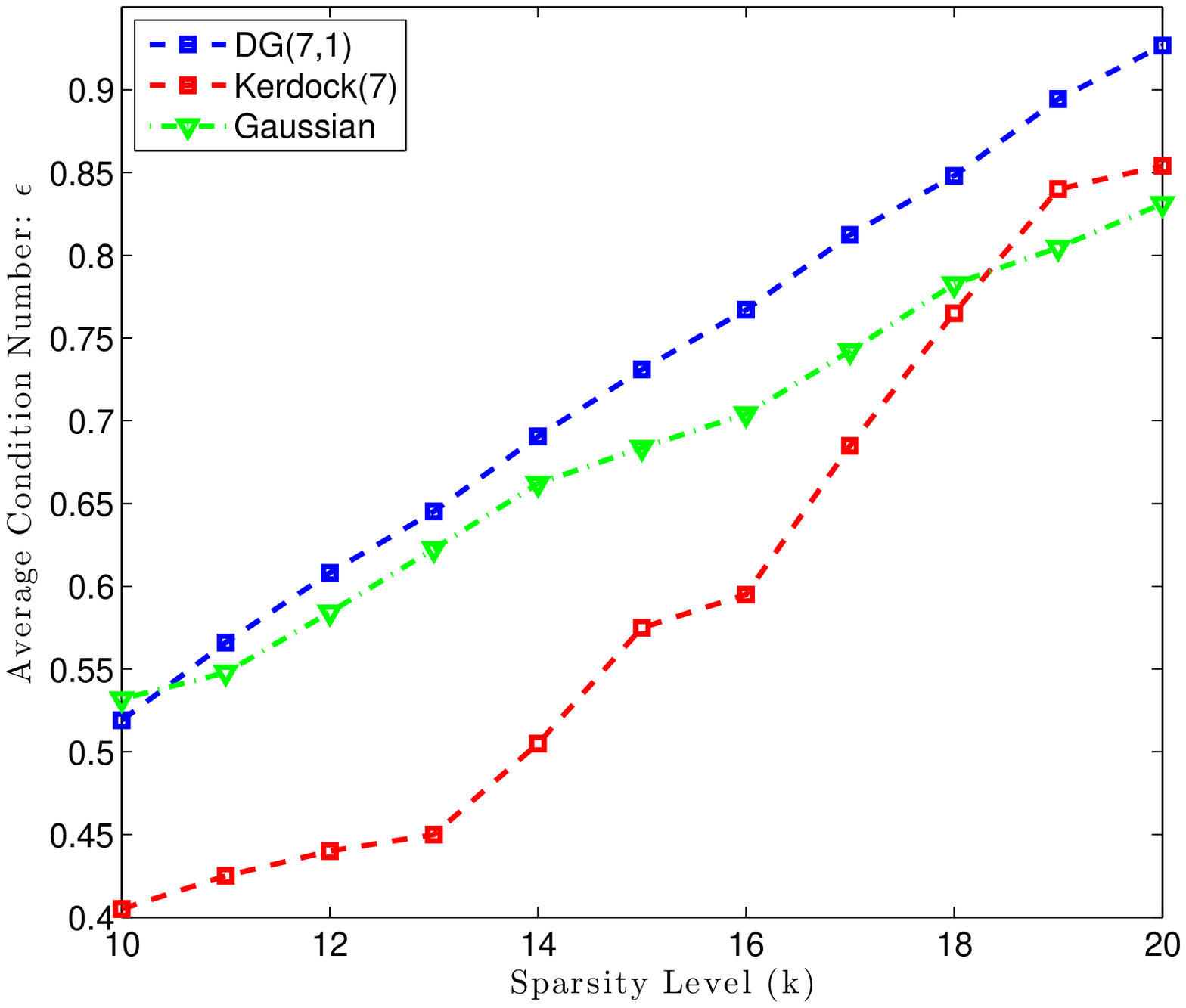}
\label{norms:fig1_1}
}
\hfil
\subfloat{\includegraphics[width=3.5in]{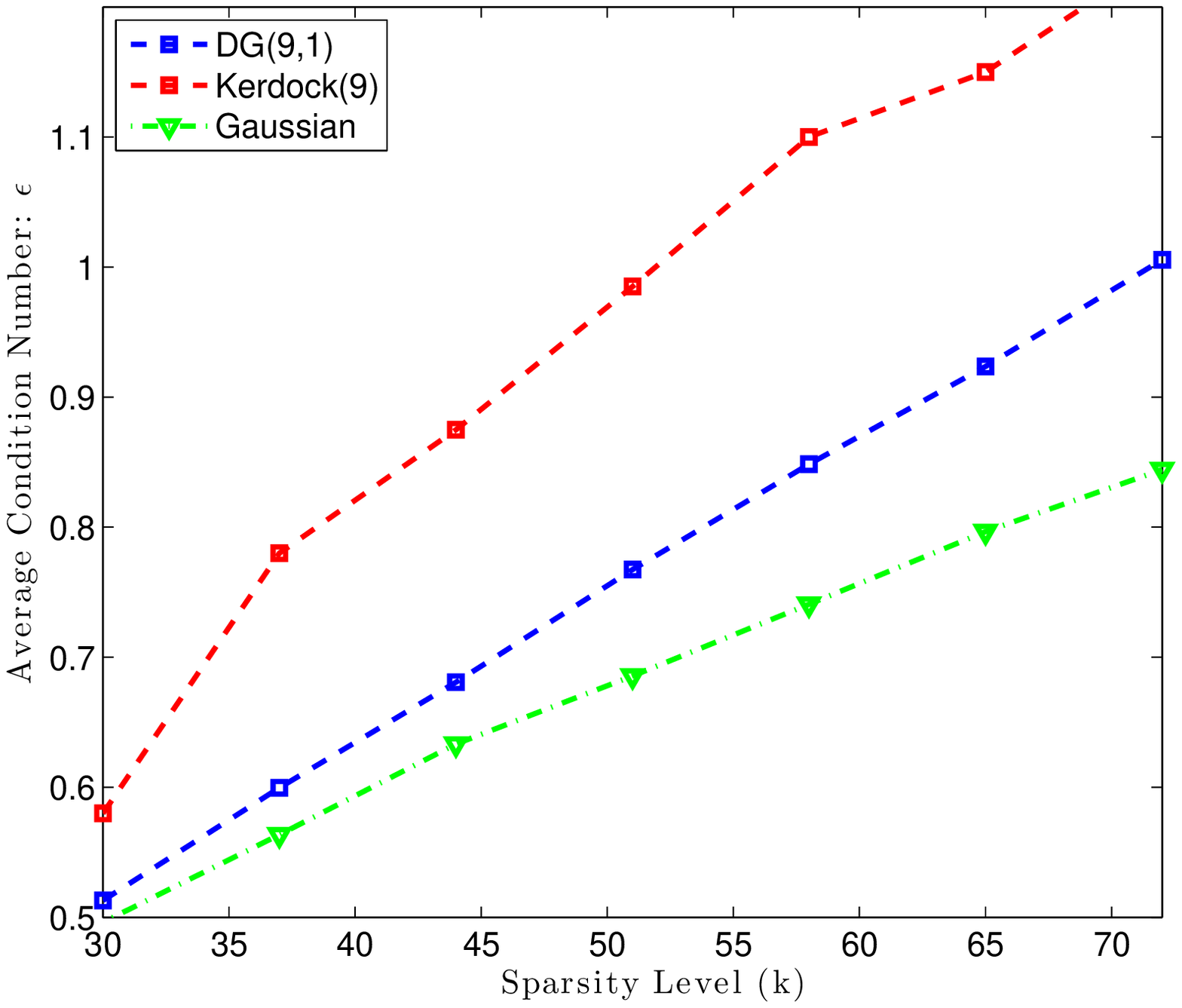}
\label{norms:fig1_2}
}
}
\caption{Average spectral norm of $\A_k^\dag\A_k-{\rm I}_{k\times k}$, where $\A_k$ is a random sub dictionary of $\A$. Here the comparison is between Gaussian, $DG(m,1)$ sieve, and $DG(m,0)$ base matrices. Each experiment is repeated $2000$ times.}
\label{norms:fig1}
\end{figure*}

\begin{remark}
{\rm
\cb{Here we compare the empirical results of Figure~\ref{norms:fig1} with the theoretical results of Theorem~\ref{bg:prop2}.  First we considered the $DG(7,0)$ frame, with $\n=2^{14}$ and $\mm=2^7$. The worst case coherence of $\A$ is $\mu = 2^{-\frac{7}{2}}$, and the square of the spectral norm of $\A$ is $2^7$. So the constant $c$ in Theorem~\ref{bg:prop3} needs to be at least $\mu\,\log \n=\frac{14\log 2}{8\sqrt{2}}\approx 0.85$. Hence, as long as $k$ is at most $\frac{0.85\,\times 128}{14 \log 2}\approx 11$, Theorem~\ref{bg:prop2} predicts probability of non-uniqueness on the order of $2^{-14}$. Experimental results presented in Figure~\ref{norms:fig1_1} are more positive; all $2000$ trials resulted in sub-dictionaries with full rank, even for k as large as $20$. 

Next we considered the $DG(7,1)$ sieve with $\n=2^{14}$ and $\mm=103$\footnote{The $25$ duplicate rows were removed from the matrix.}. The worst case coherence of $\A$ is $\mu \approx 2^{-\frac{5}{2}}$, and the square of the spectral norm of $\A$ is $\|\A\|^2\approx \frac{16384}{103}=159.6$. As a result, the constant $c$ needs to be at least $\frac{14\log 2}{4\sqrt{2}}\approx 1.70$. Therefore, as long as $k$ is less than $\frac{1.70\,\times 103}{14 \log 2}\approx 10$ Theorem~\ref{bg:prop2} predicts probability of non-uniqueness on the order of $2^{-14}$. Again, we see that the theoretical bound is not tight, and for $k$ as large as $20$ all trials provide uniqueness of sparse representation.
}
}
\end{remark}
\begin{figure*}[ht]
\centering
\fbox{ \includegraphics[width=3.5in]{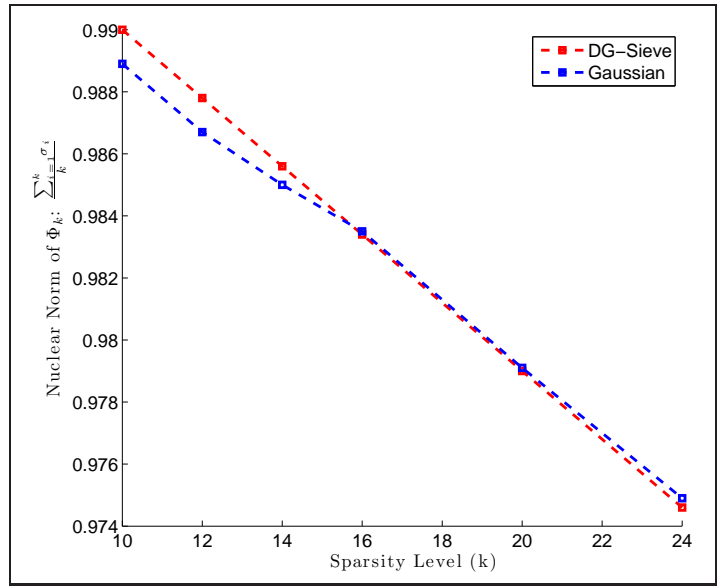}
}
\caption{Average nuclear norm $\left(\frac{1}{k}\sum_{i=1}^k \sigma_i\right)$ of random sub-dictionaries of of $DG(7,1)$ and Gaussian matrices of the same size as a function of the sparsity level $k$.}
\label{norms:fig2}
\end{figure*}

\begin{remark}
{\rm
The bounds of Proposition~\ref{bg:prop1} only apply to the condition number of random submatrices and do not provide additional information about the distribution of eigenvalues. However Gurevich and Hadani \cite{GH} have analyzed the spectrum of certain incoherent dictionaries that are unions of disjoint orthonormal bases. They have shown that the eigenvalues of the Gram matrix of a random subdictionary are asymptotically distributed around $1$ according to the Wigner semicircle law. Our experimental results suggest that this property is shared by DG sieves which are not unions of orthonormal bases. Figure~\ref{norms:fig2} shows that the distribution of the singular values of a random submatrix of a DG sieve is symmetric around 1, and very similar to the distribution for a Gaussian matrix of the same size.
}
\end{remark}

\section{Numerical Experiments}
\label{sec:exp}
\begin{figure*}[ht]
\centerline{
\subfloat[Average fraction of the support that is reconstructed successfully as a function of the sparsity level $k$]{\includegraphics[width=3.5in]{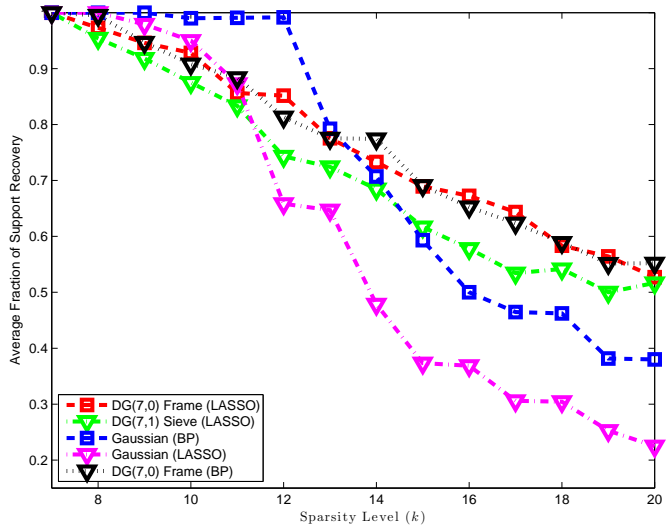}
\label{exp:fig1_1}
}
\hfil
\hskip0.2cm
\subfloat[Average reconstruction time in the noiseless regime for different  sensing matrices.]{\includegraphics[width=3.5in]{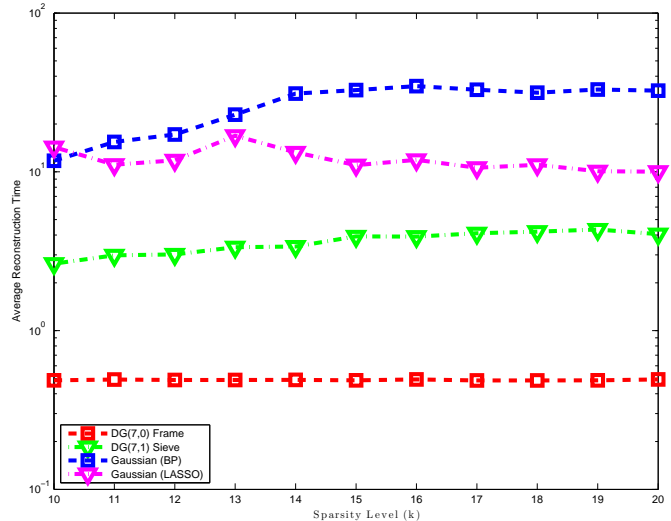}
\label{exp:fig1_2}
}
}
\caption{Comparison between $DG(7,0)$ frame, $DG(7,1)$ sieve, and Gaussian matrices of the same size in the noiseless regime. The regularization parameter for LASSO is set to $10^{-9}$.}
\label{exp:fig1}
\end{figure*}


\begin{figure*}[ht]
\centerline{
\subfloat[The impact of the noise in the measurement domain on the accuracy of the sparse approximation for different sensing matrices.]{\includegraphics[width=3.5in]{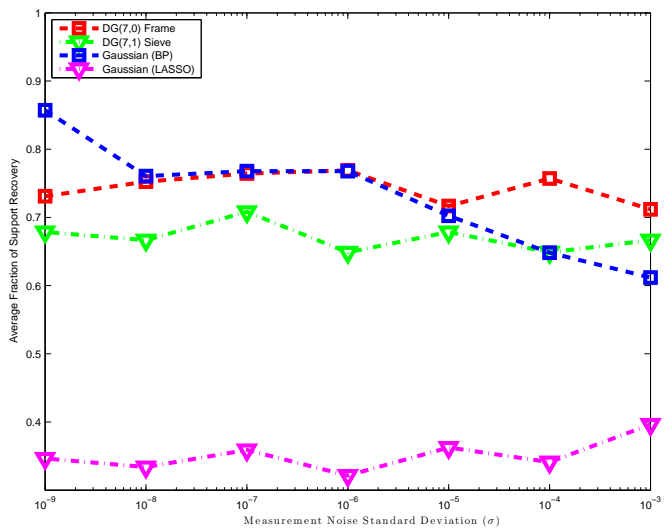}
\label{exp:fig3_1}
}
\hfil
\hskip0.2cm
\subfloat[The impact of the noise in the data domain on the accuracy of the sparse approximation for different sensing matrices. ]{\includegraphics[width=3.5in]{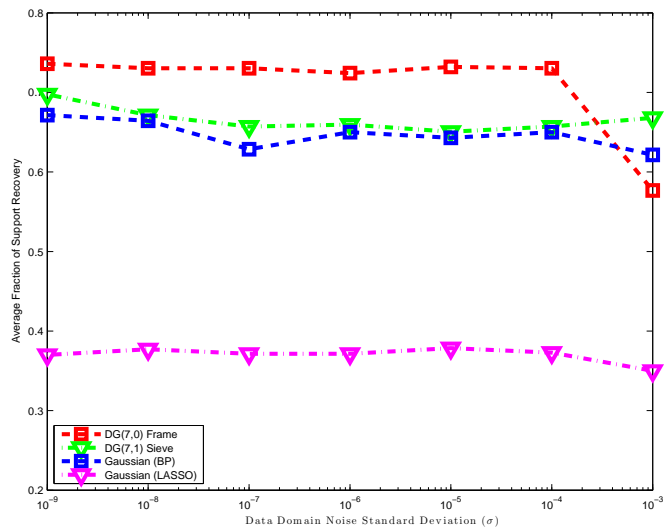}
\label{exp:fig3_2}
}
}
\caption{
Average fraction of the support that is reconstructed successfully as a function of the noise level in the measurement domain (left), and in the data domain (right). Here the sparsity level is $14$. The regularization parameter for LASSO is determined as a function of the noise variance according to Theorem~\ref{bg:prop3}.}
\label{exp:fig3}
\end{figure*}

In this Section we present numerical experiments to evaluate the performance of the DG frames and sieves. The performance of DG frames and sieves is compared with that of random Gaussian sensing matrices of the same size. The SpaRSA algorithm \cite{sparsa} with $\ell_1$ regularization parameter $\lambda = 10^{-9}$  is used for signal reconstruction in the noiseless case, and the parameter is adjusted according to Theorem~\ref{bg:prop3} in the noisy case. The reason for using SpaRSA is that is designed to solve complex valued LASSO programs.

\cb{
\begin{remark}
{\rm
Given a random sensing matrix satisfying  RIP, it is known that Basis Pursuit leads to more accurate reconstruction than the LASSO  \cite{CRT1}. It is for this reason that we also compare results for LASSO applied to DG matrices with results for Basis Pursuit applied to Gaussian matrices. The $\ell_1$-magic package\cite{magic} is used to solve the Basis Pursuit optimization program. The results for Gaussian matrices shown in Figure~\ref{exp:fig1} are consistent with the observation made in \cite{TW} that when the signal is not very sparse, interior point methods ($\ell_1$ - magic) are less sensitive than gradient descent methods (SpaRSA)
}
 \end{remark}
}

For Gaussian matrices, we sampled $10$ iid random matrices independently \cb{to eliminate the exponentially small chance of getting a sample $\A$ with $\mu=\omega\left(\mm\right)$ or $\|\A\|^2=\omega\left(\frac{\n}{\mm}\right)$}, and the median of the results among all $10$ random matrices is reported. \cb{The use of $10$ random trials to eliminate pathological sensing matrices is standard practice (see \cite{BGIST} for example).}  

\cb{
The experiments relate accuracy of sparse recovery to the sparsity level and the Signal to Noise Ratio (SNR). Accuracy is measured in terms of the statistical $0-1$ loss metric which captures the fraction of signal support that is successfully recovered. The reconstruction algorithm outputs a $k$-sparse approximation $\has$ to the $k$-sparse signal $\as$, and the statistical $0-1$ loss is the fraction of the support of $\as$ that is not recovered in $\has$. Each experiment was repeated $2000$ times and Figure~\ref{exp:fig1} records the average loss.
}

Figure~\ref{exp:fig1} plots statistical $0-1$ loss and complexity (average reconstruction time) as a function of the sparsity level $k$. We select $k$-sparse signals with uniformly random support, with random signs, and with the amplitude of non-zero entries set equal to $1$. Three different sensing matrices are compared; a Gaussian matrix, a $DG(7,0)$ frame and a $DG(7,1)$ sieve. After compressive sampling the signal support is recovered using the SpaRSA algorithm with $\lambda = 10^{-9}$. For random matrices the signal support is also recovered by $\ell_1$-minimization. 

\cb{
Figure~\ref{exp:fig3_1} plots statistical $0-1$ loss as a function of noise in the measurement domain and Figure~\ref{exp:fig3_2} does the same for noise in the data domain. In the measurement noise study, a ${\cal N}(0,\sigma^2)$ iid measurement noise vector is added to the sensed vector to obtain the $\mm$ dimensional vector $f$. The original $k$-sparse signal $\as$ is then approximated by solving the LASSO program with $\lambda=\cb{2\sqrt{2\log \n}\sigma^2}$, and basis pursuit with $\epsilon=2\mm\sigma^2$.  Following Lemma~\ref{bg:lemma1}, we use a similar method to study noise in the data domain. Figure~\ref{exp:fig3} shows that DG frames and sieves outperform random Gaussian matrices in terms of noisy signal recovery using the LASSO.
}

\section{Conclusion}
\label{sec:conc}
{ We have constructed two families of deterministic sensing matrices, $DG(m,r)$ frames and $DG(m,r)$ sieves, by exponentiating codewords from $\mathbb{Z}_4$ - linear Delsarte-Goethals codes. We have verified that the worst-case coherence and the spectral norm of these sensing matrices satisfy the conditions necessary for uniqueness of sparse representation and fidelity of $\ell_1$ reconstruction via the LASSO algorithm. We have presented numerical results that confirm performance predicted by the theory. These results show that DG frames and sieves outperform random Gaussian matrices in terms of noiseless and noisy signal recovery using the LASSO. Our focus here is on $\ell_1$ reconstruction using the LASSO algorithm but we note that the particular structure of the DG matrices leads to faster algorithms and to additional features such as local decoding and stronger guarantees on resilience to noise in the data domain.}

\section*{Acknowledgements}
The authors would like to thank Marco Duarte and Waheed Bajwa for sharing many valuable insights, and Waheed in particular for his help with the SpaRSA package.\bibliographystyle{IEEEbib}

\bibliography{lasso}
\appendices
\cb{
\section{The Number of Solutions of Condition (C1)}
\label{sec:mw}
Let $DG_0(m,r)$ denote the set of all zero-diagonal matrices in $DG(m,r)$: 
$$DG_0(m,r)=\left\{\sum_{t=1}^r P^t(a_t)\, | a_t\in \mathbb{F}_2^m\,t=1,\cdots,r \right\}.$$ 
For every matrix $P$ in $DG_0(m,r)$, the vector $xPx^\top$ is a codeword of the linear binary code $\overline{DG}_0(m,r)$ which is a sub-code of the Delsarte-Goethals code.  Note that $\overline{DG}_0(m,r)$ has $2^{rm}$ codewords of length $2^m$. The following lemma shows how the number of solutions to~(C1) is related to the properties of this binary code.
\begin{lemma}
Let $\{W_0,\cdots,W_{\mm}\}$ denote the weight distribution of $\overline{DG}_0(m,r)$. Then the number of pairs $(x,x+e)$ satisfying~(C1) is equal to
\begin{equation}
\label{mw:eq1}
\frac{1}{2^{rm}} \sum_{i=0}^{\mm} W_i {\cal K}_2(i),
\end{equation}
where ${\cal K}_{\ell}(z)$ is the ${\ell}^{th}$ Krawtchouk polynomial, defined as
\begin{equation}
\label{mw:eq2}
{\cal K}_{\ell}(z)=\sum_{r=0}^{\ell} {z \choose r} {\mm-z \choose \ell-r}  (-1)^r.
\end{equation}
\end{lemma}
\begin{IEEEproof}
Lemma~\ref{th:lemma2} implies that the number pairs $(x,x+e)$ satisfying Condition~(C1) is equal to the number of  duplicate rows in $\overline{DG}_0(m,r)$. The condition that the rows $x$ and $x+e$ are identical is equivalent to the condition that the vector with entry $1$ in positions $x$ and $x+e$, and zero elsewhere belongs to the dual code. The lemma now follows from the MacWilliams Identities \cite{MS} that relate relate the number of codewords of weight $2$ in the dual of $\overline{DG}_0(m,r)$ to the weight distribution of $\overline{DG}_0(m,r)$.
\end{IEEEproof}
Next we show that for the case $r=1$, the number of solutions to~(C1) only depends on the number of codewords with weight $2^{m-1}$ in $\overline{DG}_0(m,1)$:
\begin{theorem}
\label{mw:theorem1}
Let $m$ be an odd number and let $r$ equal $1$. Then the number of solutions to~(C1) is $2^m-1-s$ where $s$ is the number of codewords with weight $2^{m-1}$ in $\overline{DG}_0(m,1)$.
\end{theorem}
\begin{IEEEproof}
We start by calculating the rank of matrices in $DG_0(m,1)$: Let $a$ be a fixed element of $\mathbb{F}_2^m$. A field element $x$ is in the null space of $P_a$ if and only if for every field element $y$, $xP_ay^\top=0$.  Using Equation~\ref{field_eq}, this condition can be translated to the condition
$$\Tr\left((xy^2+x^2y)a\right)=0 \mbox{  for all }y.$$  Since $\Tr(x)=\Tr(x^2)$ the condition further reduces to $$\Tr\left((xa+x^4a^2)y^2\right)=0 \mbox{  for all }y.$$ Non-degeneracy of the trace implies that $x^4+\frac{x}{a}=0$, which, since $m$ is odd, has the unique solution $x^3=\frac{1}{a}$.

Now let $S=\sum_{x\in \mathbb{F}_2^m} \imath^{xP_a x^\top}$. Since $xP_a x^\top$ is a binary codeword, we have $S^2 = \left(\mm-2w_a\right)^2$, where $w_a$ is the weight of the codeword determined by $P_a$.  It has been proved in \cite{strip} that $S^2=2^m \sum_{e:eP_a=0} \imath^{eP_ae^\top}$. We provide the proof here for completeness:

We have
\begin{eqnarray}
\nonumber S^2=\sum_{x,y} i^{xP_ax^\top+yP_ay^\top}= \sum_{x,y}i^{(x+y)P_a(x+y)^\top+2{x}P_ay^\top}
\end{eqnarray}
Changing variables to $z=x\oplus y$ and $y$ gives
$$S^2=\sum_{z} i^{zP_az^\top} \sum_y(-1)^{zP_ay^\top}=2^m\sum_{z:zP_a=0}\imath^{zP_az^\top}.$$
 The null space of $P_a$ has only two elements $0$ and $a^{-\frac{1}{3}}$. As a result $$S^2=2^m\left(1+\imath^{a^{-\frac{1}{3}}P_a {a^{\frac{1}{3}}}^\top} \right).$$
There are two cases; $S^2$ is either $0$ or $2^{m+1}$.

\textbf{Case 1:} $S$ is zero. This case provides one possible weight value: $w_a=2^{m-1}$.\\
\textbf{Case 2:} $|S|^2=2^{m+1}$. Therefore $2^m-2w_a=\pm 2^{\frac{m+1}{2}}$.  This case provides two distinct weight values: $w_a=2^{m-1}\pm 2^{\frac{m-1}{2}}$. 

Hence $DG_0(m,1)$ has exactly four distinct weights $\langle 0,2^{m-1}- 2^{\frac{m-1}{2}}, 2^{m-1}, 2^{m-1}+ 2^{\frac{m-1}{2}}\rangle$. Let $\langle 1,t,s,t'\rangle$ denote the corresponding weight distribution. We can use the MacWilliams identities to find the values of $t$ and $t'$ as a function of $s$. First, note that the dual code has exactly one codeword of weight $0$. Using MacWilliams identities with Krawtchouk polynomial ${\cal K}_{0}(z)=1$, gives the equation $1+t+s+t'=\n$. Second, since all matrices in $DG_0(m,r)$ are zero-diagonal, for every field element $a$ and for every index $j$ in $\{0,\cdots,m\}$, $\xi^j P_a {\xi^j}^\top=0$, the dual code has exactly $m+1$ codewords of weight $1$. Again,  MacWilliams identities, with Krawtchouk polynomial ${\cal K}_{1}(z)=\mm-2z$ gives the equation  $(m+1)\mm=\mm+\sqrt{2\mm}(t'-t)$. This equation can be simplified to $t-t'=m\,2^{\frac{m-1}{2}}$. Solving $t$ and $t'$ with respect to $s$ gives $t=\frac{2^m-1-s+m2^{\frac{m-1}{2}}}{2}$ and $t'=\frac{2^m-1-s-m2^{\frac{m-1}{2}}}{2}$. The theorem then follows from substituting  the values $t,s,t'$ into Equation~\eqref{mw:eq2}, and simplifying the expression using the Krawtchouk polynomial ${\cal K}_{2}(z)=\frac{\left(\mm-2z\right)^2-\mm}{2}$.
\end{IEEEproof}

}
\end{document}